\documentclass[10pt,onecolumn]{article}
\usepackage[margin=1in,nohead]{geometry}
\usepackage{cite}
\usepackage{microtype}
\usepackage{amsmath}
\usepackage{amssymb}
\usepackage{amsthm}
\usepackage{fix-cm}
\usepackage{mathtools}
\usepackage{algorithm}
\usepackage[noend]{algpseudocode}
\usepackage{graphicx}
\usepackage{epsfig}
\usepackage[font=footnotesize]{caption}
\usepackage{siunitx}
\usepackage{subcaption}
\usepackage{tabularx}
\usepackage{booktabs}
\usepackage{multirow}
\usepackage{makecell}
\usepackage[hyphens]{url}
\usepackage[hyphenbreaks]{breakurl}
\usepackage{enumitem}
\usepackage{tikz}
\newtheorem{theorem}{Theorem}

\newtheorem{proposition}[theorem]{Proposition}

\usetikzlibrary{arrows,calc,positioning,trees,shapes,matrix,backgrounds}
\pgfdeclarelayer{background}
\pgfsetlayers{background,main}

\renewcommand{\square}{\lozenge}
\newcommand{\tablefontsize}{\footnotesize}
\newcommand{\defbox}[2]{
\vspace{1mm}
\noindent\fbox{
\begin{minipage}{0.96\textwidth}
{\bf #1}
#2
\end{minipage}
}
\vspace{1mm}
}

\title{Simple and efficient four-cycle counting on sparse graphs}
\author{Paul Burkhardt\footnote{Research Directorate, National Security Agency,
Fort~Meade, MD 20755. Email: pburkha@nsa.gov}~~and
David G. Harris\footnote{University of Maryland, Department of Computer
    Science, College Park, MD 20742. Email: davidgharris29@gmail.com}}

\date{\today}

\begin{document}

\maketitle

\begin{abstract}
  We consider the problem of counting 4-cycles ($C_4$) in an undirected graph
  $G$ of $n$ vertices and $m$ edges (in bipartite graphs, 4-cycles are also
  often referred to as \emph{butterflies}). Most recently, Wang et al. (2019,
  2022) developed algorithms for this problem based on hash tables and sorting
  the graph by degree. Their algorithm takes $O(m\bar\delta)$ expected time and
  $O(m)$ space, where $\bar \delta \leq O(\sqrt{m})$ is the \emph{average
  degeneracy} parameter introduced by Burkhardt, Faber \& Harris (2020). We
  develop a streamlined version of this algorithm requiring $O(m\bar\delta)$
  time and precisely $n$ words of space. It has several practical improvements
  and optimizations; for example, it is fully deterministic, does not require
  any auxiliary storage or sorting of the input graph, and uses only addition
  and array access in its inner loops.
  
  Our algorithm is very simple and easily adapted to count 4-cycles incident to
  each vertex and edge. Empirical tests demonstrate that our array-based
  approach is $4\times$ -- $7\times$ faster on average compared to popular hash
  table implementations.

  \textit{Keywords}: graph, 4-cycle, four-cycle, counting, butterfly
\end{abstract}

\section{\label{sec:introduction}Introduction}
Let $G = (V,E)$ be an undirected graph, with $n = |V|$ vertices and $m = |E|$
edges. A $k$-cycle in $G$, denoted $C_k$, is a sequence of $k$ distinct adjacent
edges that begins and ends with the same vertex without repeating
vertices. Cycles are fundamental building blocks to larger graph structures and
often denote important properties. The $C_3$ (i.e. triangle) in particular has
been well-studied for its role in connectivity, clustering, and
centrality~\cite{bib:watts_strogatz1998, bib:cohen2009, bib:friggeri2011,
  bib:burkhardt2021b}. Today there is interest in analyzing fundamental 4-node
subgraph patterns in large, sparse graphs. (We say a graph is \emph{sparse} if
it has $m=o(n^2)$ edges.)

The complexity of graph algorithms is often measured in terms of $m$ and
$n$. For example, a straightforward triangle-counting algorithm may have
$O(m^{3/2})$ runtime. Such super-linear bounds are impractical on large-scale
graphs. More carefully designed algorithms take advantage of the graph structure
and thereby require a finer measure of the runtime, often in terms of
``sparsity'' parameters. One such parameter is the \emph{core number} $\kappa$,
which is the largest minimum degree of any subgraph. This is closely related to
the \emph{arboricity} $a$, which is the minimum number of forests into which the
edges can be partitioned. An alternative measure related to the arboricity is a
new parameter called the \emph{average degeneracy}, introduced in 2020 by
Burkhardt, Faber and Harris~\cite{bib:burkhardt2020}. It is particularly
powerful and defined as
\[
\bar\delta = \frac{1}{m}\sum_{uv \in E} \min\{d(u),d(v)\},
\]
where $d(x)$ is the degree of any vertex $x$. It is known that
$\bar\delta\leq 2a$ and $a \le \kappa< 2a\le O(\sqrt{m})$, and
$\bar\delta$ can be much smaller than $a$ or $\kappa$. In our
experiments on real-world graphs we have found that on average $\bar\delta$
is about $3\times$ smaller than $\kappa$, and on some graphs it was $6\times$
smaller.

There has been interest in counting 4-cycles in bipartite graphs as these can
suggest similarity or closeness between vertices~\cite{bib:sanei2018,
  bib:wang2019, bib:zhou2021, bib:wang2022}. In this context, a $C_4$ is often
called a \emph{butterfly.} For example, a $C_4$ in a user-item bipartite graph
represents groups of users that select the same subset of items, and may
therefore have similar interests. Recently, Wang et al.~\cite{bib:wang2019,
  bib:wang2022} developed an algorithm with $O(m\bar\delta)$ expected runtime
and $O(m)$ space based on sorting vertices. The algorithm of~\cite{bib:wang2019,
  bib:wang2022} was designed for bipartite graphs, but it works for general
sparse graphs. This is the current state of the art for efficient $C_4$
counting.

In this paper, we present a streamlined version of the \cite{bib:wang2022}
algorithm for $C_4$-counting, based on simple arrays with a number of key
optimizations. These include some general algorithmic tricks which may not be
obvious to non-experts. Other optimizations involve some time-memory tradeoffs
where empirical evidence is needed to choose the best implementation. Our main
improvements to the state-of-the-art are (i) asymptotically less space and (ii)
deterministic runtime. Our experiments also show that using an array is
$4\times$ -- $7\times$ faster than using popular hash tables. We summarize our
results as follows.

\begin{theorem}
  \label{thm:c4}
  Counting 4-cycles in a graph $G$ takes $O(m\bar\delta)$ time and $n$ space.
\end{theorem}

\begin{theorem}
  \label{thm:c4vertex}
  Counting 4-cycles incident to each vertex in a graph $G$ takes
  $O(m\bar\delta)$ time and $O(n)$ space.
\end{theorem}

\begin{theorem}
  \label{thm:c4edge}
  Counting 4-cycles incident to each edge in a graph $G$ takes $O(m\bar\delta)$
  time and $O(m+n)$ space.
\end{theorem}

\begin{theorem}
  \label{thm:c4list}
  Enumerating 4-cycles in a graph $G$ takes $O(m\bar\delta + \#C_4)$ time and
  $O(n)$ space.
\end{theorem}

We emphasize that $G$ can be provided with an \emph{arbitrary} ordering of the
neighbors of each vertex. These asymptotic bounds are similar to the previous
algorithms but there are two main differences. First, they are completely
deterministic; in particular, the algorithms avoid the use of hash
tables. Second, the space complexity is just linear in the vertex count
(ignoring the cost of storing the input graph). To the best of our knowledge,
all previous efficient algorithms for $C_4$ counting have required $\Omega(m)$
additional space.

These differences may seem minor but have important practical consequences. Let
us summarize some other key reasons these differences can be significant.

\begin{description}
\item[Determinism vs. randomness:] The hash table is a powerful, versatile data
  structure which is ubiquitous in algorithm design. It has essentially optimal
  memory usage, linear in the size of its data, as well as constant access time
  on average. However, in addition to theoretical issues about the role of
  randomness, it has its downsides in practice. First, it is relatively costly
  to implement, requiring a good hash function and code to handle
  collisions. Hash tables using chaining are not cache friendly, and those based
  on probing require careful maintenance of the load factor. Many programming
  languages provide hash table implementations for convenience but are of
  varying quality. We discuss this issue further in
  Section~\ref{sec:experiment}, where we demonstrate empirically that hash
  tables may be an order of magnitude slower than direct table access.
  
  In addition, like all randomized algorithms, hash table operations can have
  unpredictable or unreproducible behavior. More seriously, it can open the door
  to catastrophically bad performance on adversarial
  inputs~\cite{bib:crosby2003denial}. This can be partially mitigated by using
  stronger hash functions, e.g.~\cite{bib:bernstein}. By using a deterministic
  algorithm with \emph{worst-case} guarantees, we eliminate all such issues at a
  stroke.

\item[Cost of graph storage:] In some applications, we may wish to run our
  algorithm on very large inputs using a single machine~\cite{bib:park2021,
    bib:mosaic2017, bib:garaph2017, bib:khaouid2015, bib:graphchi2012}. The
  graph may require hundreds of gigabytes to store, and we can afford to keep at
  most one copy in memory. (By contrast, the number of \emph{vertices} is
  typically more modest, requiring only megabytes of storage.) Any storage
  overhead of $\omega(m)$ is out of the question. Storage overheads of
  $\Omega(m)$ are painful and best avoided as much as possible. Note that many
  previous $C_4$-counting algorithms may have hidden $\Omega(m)$ costs for
  scratch-space to sort the input graph by vertex degree.

\item[Read-only graph access:] Because of the large size of the input graph, it
  is often desired to use ``graph packages,'' which are responsible for
  maintaining and updating the graph and communicating it across a network; for
  example, see~\cite{bib:stinger}. It may be running multiple graph algorithms
  simultaneously (e.g. 4-cycle counting, $k$-core extraction, 3-cycle
  counting). In these cases, the individual graph algorithms should not modify
  the graph storage itself.
\end{description}

There is one other algorithmic improvement worth mentioning. In our algorithms,
we do not explicitly compute binomial coefficients. Instead, we update the
counts in a specific order leading to a sum equating to the binomial
coefficient. In our experiments we found this slightly improved the running
time. This also eliminates numerical complications from multiplication, which
could potentially overflow or lose precision. Notably, this technique would also
simplify and improve the current hash-based methods~\cite{bib:sanei2018,
  bib:wang2019, bib:zhou2021, bib:wang2022} by eliminating the second pass over
the graph.

For clarity, we have tried throughout to present our algorithm in a single
listing, with only basic and explicit data structures.

\subsection{\label{sec:background}Related algorithms}
When considering 4-cycles, one must be careful to distinguish between algorithms
that \emph{detect} 4-cycles, \emph{count} 4-cycles, and \emph{enumerate}
4-cycles. A classical combinatorial result~\cite{bib:bondy} is that graphs with
many edges are guaranteed to contain many 4-cycles. Hence, 4-cycle detection is
an inherently simpler problem and algorithms for it can focus only on highly
sparse graphs. For instance, there is a combinatorial algorithm
of~\cite{bib:ayz1994} to detect $C_4$ in $O(m^{4/3})$ time. Note that this is
completely different from the situation for 3-cycles, where there exist
triangle-free graphs with $\Theta(n^2)$ edges.

For the problem of enumerating 4-cycles, the runtime typically should be linear
in $\#C_4$ (the number of copies of $C_4$ in the graph). Again, such algorithms
can focus on sparse graphs, as otherwise $\#C_4$ is guaranteed to be large.
Recently, Abboud et al.~\cite{bib:abboud_c4listing} proposed an algorithm to
enumerate 4-cycles with runtime $\tilde{O}( \min\{ n^2, m^{4/3} \} + \# C_4)$
time and $\tilde{O}(\min\{ n^2, m^{4/3} \})$ space (where $\tilde{O}(\cdot)$
hides logarithmic terms). The algorithm of~\cite{bib:wang2022} can also
enumerate 4-cycles with $O( m \bar \delta + \# C_4) \leq O(m^{3/2} + \# C_4)$
time and $O(m)$ space. There is also evidence that runtime $\Omega( \min\{ n^2,
m^{4/3} \} + \# C_4)$ is necessary for 4-cycle
enumeration~\cite{bib:jin2023removing,bib:abboud2023stronger}.

For the harder problem of 4-cycle \emph{counting}, there are two main
approaches. The first of which are algorithms based on fast matrix
multiplication~\cite{bib:ayz1994, bib:ayz1995, bib:ayz1997}. The best result in
general for 4-cycle counting takes $n^{\omega+o(1)}$ time~\cite{bib:ayz1997},
where $\omega$ is the matrix multiplication exponent. For sparse graphs,
Williams et al. described a randomized algorithm with
$m^{\frac{4\omega-1}{2\omega+1}+o(1)}$ runtime~\cite{bib:williams2015}, which is
$O(m^{1.48})$ given current bounds $\omega \leq
2.37286$~\cite{bib:alman2021}. We emphasize that such algorithms are not
practical for three distinct reasons: (i) For large sparse graphs, we desire
runtimes that are nearly linear in $m$, which may depend on graph sparsity
parameters such as $\bar \delta$; (ii) fast matrix multiplication inherently
requires super-linear storage for the recursive matrix subdivisions; (iii) fast
matrix multiplication algorithms, with the possible exception of Strassen's
algorithm, have very high constant factors and are generally considered
impractical.

A second approach is to use \emph{combinatorial} algorithms to count
4-cycles. This is the approach we adopt. These algorithms can take advantage of
various types of graph sparsity; in addition, since they are based on simple
practical data structures, the hidden constants are small and reasonable. For
example, the algorithm of~\cite{bib:sanei2018} runs in $O\left(\min \{
\sum_{u\in L} d(u)^2, \sum_{v\in R} d(v)^2 \} \right)$ time and $O(n)$ space for
a bipartite graph $G = (L \cup R, E)$. See also
\cite{bib:hovcevar2017,bib:pinar2017escape} for similar algorithms.

There are also many related algorithms for the simpler task of counting 3-cycles
(i.e. triangles). Since triangle-free graphs may have $\Theta(n^2)$ edges,
triangle-counting, triangle-enumeration, and triangle-detection tend to have
similar algorithmic behavior. In particular, a sparse graph can have
$\Theta(m^{3/2})$ triangles, and an optimal triangle-enumeration algorithm
should take $O(m^{3/2})$ time. Tighter bounds were given by Chiba and
Nishizeki~\cite{bib:chiba_nishizeki1985}; they gave an efficient, deterministic
algorithm to enumerate triangles with $O(ma)$ runtime where $a$ is the
arboricity of $G$. They also provided a deterministic algorithm to count
4-cycles with $O(ma)$ runtime and $O(m)$ space. There are also algorithms based
on fast matrix multiplication to count 3-cycles in $m^{\frac{2\omega}{\omega+1}
  + o(1)} \leq O(m^{1.41})$ time.

\subsection{\label{sec:notation}Notation}
The vertices of the graph are assumed to have labels in $[n] = \{1, \dots, n
\}$. We use $uv$ as shorthand for an edge $\{u,v\}$ when convenient. The
neighborhood of a vertex $v$ is given by $N(v)=\{u: uv\in E\}$ and $d(v)=\lvert
N(v) \rvert$ is its degree. We define a total order of $V$ by setting $u \prec
v$ if $d(u) < d(v)$; if $d(u) = d(v)$, we break ties arbitrarily (e.g. by vertex
ID).

A simple cycle of length $k$ is denoted by $C_k$. Note that any 4-cycle
$x,y,z,w$ has eight related 4-cycles, which can be obtained by shifting the
sequence (e.g. $y,z,w,x$), or by reversing the sequence (e.g. $x,w,z,y$). In
counting \emph{distinct} 4-cycles, we would count all eight such cycles as just
a single $C_4$. We use $\square(G)$ for the total number of distinct 4-cycles in
$G$ and likewise $\square(v)$ and $\square(e)$, respectively, for the number of
distinct 4-cycles involving a vertex $v$ and an edge $e$.

\section{\label{sec:main}Principal algorithm}
Our principal solution for counting 4-cycles is given in
Algorithm~\ref{alg:c4}. We emphasize that the algorithm does not require
explicit ordering of vertices in the graph data structure. We can determine
$u\prec v$ during the run in constant-time, by looking up the degrees of $u$ and
$v$. In Section~\ref{sec:opt} we describe how this algorithm can be optimized if
we pre-sort the data structure by degree order, which eliminates many
conditional checks.

\begin{algorithm}[H]
  \caption{4-cycle counting}\label{alg:c4}
  \begin{algorithmic}[1]
    \Require zero-initialized array $L$ of size $n$ indexed by $V$
    \For{$v\in V$}
    \label{algc4:start}
      \For{$u\in N(v)$}
        \If{$u\prec v$}
        \label{algc4:uprecv}
          \For{$y\in N(u)$}
            \If{$y\prec v$}
            \label{algc4:work_start}
              \State set $\square(G)\coloneqq \square(G) + L(y)$
              \State set $L(y) \coloneqq L(y) + 1$
            \EndIf
            \label{algc4:work_end}
          \EndFor
        \EndIf
      \EndFor
      \For{$u\in N(v)$}
      \label{algc4:2ndpass_start}
        \If{$u\prec v$}
          \For{$y\in N(u)$}
            \State set $L(y) \coloneqq 0$
          \EndFor
        \EndIf
      \EndFor
      \label{algc4:2ndpass_end}
    \EndFor
    \label{algc4:end}
  \end{algorithmic}
\end{algorithm}

To ensure that each $C_4$ is counted exactly once, we require $v$ to have higher
degree (or label in the case of tie) than $u,w$, and $y$. The array $L$ is used
to count the number of vertices $u \in N(v)$ where $u \prec v$ and $u$ itself
has a neighbor $y\prec v$. There are $\binom{L(y)}{2}$ unique neighbor pairs
$u,u'$ that give a $C_4$ of the form $\{v, u, y, u'\}$. To avoid spurious
counting from subsequent starting vertices, the second pass at
Lines~\ref{algc4:2ndpass_start}~--~\ref{algc4:2ndpass_end} clears the $L$
table. These steps are repeated for every vertex in $G$.

Note that Algorithm~\ref{alg:c4} does not explicitly compute the binomial
coefficient; instead it is implicitly computed via the identity
$\sum_{i}^{L(y)-1} i = \binom{L(y)}{2}$. This observation can be used to obtain
a particularly elegant variant of Algorithm~\ref{alg:c4} using hash tables in
place of the $L$ array, which eliminates the second pass
(Lines~\ref{algc4:2ndpass_start}~--~\ref{algc4:2ndpass_end}). See
Appendix~\ref{app:hash} for further details.

Before the formal analysis, let us consider a few simple cases for
illustration. Consider a $C_4$ as shown in Figure~\ref{fig:c4}, where we suppose
that $1 \prec 2 \prec 3 \prec 4$.

\begin{figure}[h]
  \centering
  \begin{tikzpicture}
    [
    vertex/.style={circle,draw=black,thin,node distance=0.75,inner sep=1.75},
    edge/.style={thin}
    ]
    \node [vertex] (1) {1};
    \node [vertex] (2) [right=of 1] {2};
    \node [vertex] (3) [below=of 2] {3};
    \node [vertex] (4) [below=of 1] {4};
    \draw [edge] (1) to (2);
    \draw [edge] (1) to (4);
    \draw [edge] (2) to (3);
    \draw [edge] (3) to (4);
  \end{tikzpicture}
  \caption{A simple cycle of length 4 ($C_4$).}\label{fig:c4}
\end{figure}

We begin in order from vertex $1$ in sequence to vertex $4$. Vertex $1$ is
skipped because its neighbors do not precede it in degree or label. From vertex
$2$ we get its neighbor $1$ but no updates can be made to $L$ because $4$ does
not precede $2$. At vertex $3$ only the $L(1)$ value can be incremented, due to
vertex $2$, but since the update to $L(1)$ is after the update to $\square(G)$,
then no 4-cycle can be detected. The $L(1)$ value is reset to zero before
proceeding to vertex $4$ to avoid spurious counting. Finally, from $4$ we find
$L(2)$ is incremented via $1$ and $3$ leading to the detection of the
anticipated cycle.

It should be noted that Algorithm~\ref{alg:c4} reports the count only and not
the explicit $C_4$. But, for the purpose of this discussion, we uniquely label a
$C_4$ by the ordering $(v, u, y, u')$ where $v$ is the vertex from which the
cycle was detected, $u,u'$ are neighbors of $v$ where $u<u'$, and $y$ is the
common neighbor of $u$ and $u'$. We remark that the order $u<u'$ is arbitrary as
there are two directed walks starting and ending from any vertex on a 4-cycle;
choosing $u'<u$ gives the counter-oriented walk. Tracing our algorithm on the
$C_4$ in Figure~\ref{fig:c4} implies $(4, 1, 2, 3)$, which denotes the walk
$4\rightarrow 1\rightarrow 2\rightarrow 3\rightarrow 4$. The counter-oriented
walk is simply an exchange in order between $1$ and $3$. Let us consider the
4-clique and its corresponding $C_4$'s given in Figure~\ref{fig:4clique}.

\begin{figure}[H]
  \centering
  \begin{subfigure}[t]{.2\columnwidth}
    \centering
    \begin{tikzpicture}
      [
      vertex/.style={circle,draw=black,thin,node distance=0.75,inner sep=1.75},
      edge/.style={thin}
      ]
      \node [vertex] (1) {1};
      \node [vertex] (2) [right=of 1] {2};
      \node [vertex] (3) [below=of 2] {3};
      \node [vertex] (4) [below=of 1] {4};
      \draw [edge] (1) to (2);
      \draw [edge] (1) to (4);
      \draw [edge] (2) to (3);
      \draw [edge] (3) to (4);
      \draw [edge] (1) to (3);
      \draw [edge] (2) to (4);
    \end{tikzpicture}
  \end{subfigure}
  \begin{subfigure}[t]{.2\columnwidth}
    \centering
    \begin{tikzpicture}
      [
      vertex/.style={circle,draw=black,thin,node distance=0.75,inner sep=1.75},
      edge/.style={thin}
      ]
      \node [vertex] (1) {1};
      \node [vertex] (2) [right=of 1] {2};
      \node [vertex] (3) [below=of 2] {3};
      \node [vertex] (4) [below=of 1] {4};
      \draw [edge] (1) to (4);
      \draw [edge] (2) to (3);
      \draw [edge] (1) to (3);
      \draw [edge] (2) to (4);
    \end{tikzpicture}
  \end{subfigure}
  \begin{subfigure}[t]{.2\columnwidth}
    \centering
    \begin{tikzpicture}
      [
      vertex/.style={circle,draw=black,thin,node distance=0.75,inner sep=1.75},
      edge/.style={thin}
      ]
      \node [vertex] (1) {1};
      \node [vertex] (2) [right=of 1] {2};
      \node [vertex] (3) [below=of 2] {3};
      \node [vertex] (4) [below=of 1] {4};
      \draw [edge] (1) to (2);
      \draw [edge] (3) to (4);
      \draw [edge] (1) to (3);
      \draw [edge] (2) to (4);
    \end{tikzpicture}
  \end{subfigure}
  \begin{subfigure}[t]{.2\columnwidth}
    \centering
    \begin{tikzpicture}
      [
      vertex/.style={circle,draw=black,thin,node distance=0.75,inner sep=1.75},
      edge/.style={thin}
      ]
      \node [vertex] (1) {1};
      \node [vertex] (2) [right=of 1] {2};
      \node [vertex] (3) [below=of 2] {3};
      \node [vertex] (4) [below=of 1] {4};
      \draw [edge] (1) to (2);
      \draw [edge] (1) to (4);
      \draw [edge] (2) to (3);
      \draw [edge] (3) to (4);
    \end{tikzpicture}
  \end{subfigure}
  \caption{A 4-clique and its unique $C_4$'s.}\label{fig:4clique} 
\end{figure}

There are three distinct copies of $C_4$ in a 4-clique.\footnote{There are
$\frac{(k-1)!}{2} \binom{n}{k}$ distinct copies of $C_k$ in an $n$-vertex
clique.} We'll proceed again sequentially from vertices $1$ to $4$. As before,
vertex $1$ is skipped due to degree/label ordering. From $2$ there can be no
updates to $L$ because neighbors of $1$ do not precede $2$. At $3$ the value
$L(2)$ is updated from $1$ and $L(1)$ is updated from $2$, but neither meet the
minimum count of two for a cycle. After zeroing out $L(1)$ and $L(2)$ we
continue to vertex $4$. At vertex $4$ we find $L(1)$, $L(2)$, $L(3)$ all have
value $2$, corresponding to the 4-cycles $(4,2,1,3)$, $(4,1,2,3)$, $(4,1,3,2)$
respectively.

Let us now formally prove the algorithm properties asserted by
Theorem~\ref{thm:c4}.

\setcounter{theorem}{0}
\begin{proposition}
  \label{clm:arbor}
  Algorithm~\ref{alg:c4} takes $O(m \bar \delta)$ work and $n$ space.
\end{proposition}

\begin{proof}
  The space usage is due to storing the array $L$.
  
  The work for Lines~\ref{algc4:start}~--~\ref{algc4:uprecv} is at most $\sum_{v
    \in V} \sum_{u \in N(v)} O(1) = O(m)$.

  For a given pair of vertices $v,y$,
  Lines~\ref{algc4:work_start}~--~\ref{algc4:work_end} take $O(1)$ time. So, the
  total work for Lines~\ref{algc4:work_start}~--~\ref{algc4:work_end} is given
  by $$ \sum_{v \in V} \sum_{\substack{u \in N(v) \\ u \prec v}} \sum_{
    \substack{y \in N(u)}} O(1) = \sum_{\substack{uv \in E \\ u \prec v}} O(
  d(u) )
    $$ Note that if $u \prec v$ then $d(u) \leq d(v)$. So this is at most
  $\sum_{uv \in E} \min \{ d(u), d(v) \} = m \bar \delta$ by definition. The
  analysis of Lines~\ref{algc4:2ndpass_start}~--~\ref{algc4:2ndpass_end} is
  completely analogous.
\end{proof}

\begin{proposition}
  \label{prop:correct}
  Algorithm~\ref{alg:c4} gives the correct count of $\square(G)$.
\end{proposition}

\begin{proof}
  Each distinct $C_4$ can be uniquely written by the tuple $(v,u,y,u')$ where
  $u\prec u'\prec v$ and $y\prec v$. We claim that the algorithm counts each
  such 4-tuple exactly once, namely, in the loop on $v$. For fixed $v,y$, note
  that, at the end of the iteration, the final value $\ell = L(y)$ counts the
  number of common neighbors in $N(v) \cap N(y)$ which precede $v$. In
  particular there are $\binom{\ell}{2}$ unique combinations of neighbor pairs
  $u,u' \in N(v) \cap N(y)$ with $u \prec u' \prec v$. Furthermore, the
  algorithm will increment $\square(G)$ by $\sum_{i=0}^{\ell-1} i =
  \binom{\ell}{2}$.
\end{proof}

A straightforward modification of Algorithm~\ref{alg:c4} can be used to list all
copies of $C_4$. We modify $L(y)$ to hold a linked-list, instead of just a
count, for each path $(v,u,y)$ with $u\prec v$ and $y\prec v$. See
Algorithm~\ref{alg:c4list}.

\begin{algorithm}[H]
  \caption{4-cycle enumeration}\label{alg:c4list}
  \begin{algorithmic}[1]
    \Require linked-list $L(v)$ for each vertex $v \in V$, all initialized to be
    empty
    \For{$v\in V$}
      \For{$u\in N(v)$}
        \If{$u\prec v$}
          \For{$y\in N(u)$}
            \If{$y\prec v$}
              \For{each vertex $x\in L(y)$}
              \label{algc4list:output_start}
                \State output $(v,u,y,x)$
              \EndFor
              \label{algc4list:output_end}
              \State add $u$ to end of list $L(y)$
            \EndIf
          \EndFor
        \EndIf
      \EndFor
      \For{$u\in N(v)$}
        \If{$u\prec v$}
          \For{$y\in N(u)$}
            \State set $L(y) \coloneqq \emptyset$
          \EndFor
        \EndIf
      \EndFor
    \EndFor
  \end{algorithmic}
\end{algorithm}

The loop at Lines~\ref{algc4list:output_start}--\ref{algc4list:output_end}
clearly takes $O(\square(G))$ time (as each iteration outputs a new $C_4$). The
other steps are as in Algorithm~\ref{alg:c4}. This proves
Theorem~\ref{thm:c4list}.

\section{\label{sec:completecount} Local 4-cycle counts} 
Formally, the local 4-cycle counts $\square(v)$ and $\square(e)$ respectively
denote the number of $C_4$'s containing vertex $v$ and edge $e$. Note that,
given the $\square(v)$ counts, we can recover $\square(G)$ from the simple
identity:

$$
\square(G) = \frac{1}{4} \sum_{v \in V} \square(v)
$$
and likewise given the $\square(e)$ counts, we can recover $\square(G)$
and $\square(v)$ from the identities

$$
\square(v) = \frac{1}{2} \sum_{u \in N(v)} \square( uv ), \qquad \square(G) =
\frac{1}{4} \sum_{e \in E} \square(e).
$$

Algorithm~\ref{alg:c4} can be modified to output the local count for each
vertex, as shown in Algorithm~\ref{alg:c4vertex}:

\begin{algorithm}[H]
  \caption{Vertex-local 4-cycle counting}\label{alg:c4vertex}
  \begin{algorithmic}[1]
    \Require zero-initialized arrays $L$ and $\square$ of size $n$ indexed by
    $V$, where $L$ holds 2-tuples.
    \For{$v\in V$}
      \For{$u\in N(v)$}
        \If{$u\prec v$}
          \For{$y\in N(u)$}
            \If{$y\prec v$}
              \State set $\square(v) := \square(v) + L(y).\text{orig}$
              \label{algc4vertex:1stform}
              \State set $\square(y) := \square(y) + L(y).\text{orig}$
              \label{algc4vertex:2ndform}
              \State set $L(y).\text{copy} \coloneqq L(y).\text{orig}$
              \label{algc4vertex:g1}
              \Comment{save original $L(y)$ value}
              \State set $L(y).\text{orig} \coloneqq L(y).\text{orig} + 1$
              \label{algc4vertex:2pathend}
            \EndIf
          \EndFor
        \EndIf
      \EndFor
      \For{$u\in N(v)$}
      \label{algc4vertex:2ndpass_start}
        \If{$u\prec v$}
          \For{$y\in N(u)$}
            \If{$y\prec v$}
              \State set $\square(u)\coloneqq \square(u) + L(y).\text{copy}$
              \label{algc4vertex:3rdform}
              \State set $L(y).\text{orig} \coloneqq 0$
            \EndIf
          \EndFor
        \EndIf
      \EndFor
    \EndFor
    \label{algc4vertex:2ndpass_end}
  \end{algorithmic}
\end{algorithm}

Recall in our main algorithm that a unique 4-cycle starting from $v$ is detected
on updates to $L(y)$. Thus we increment $\square(v),\square(y)$ in this same
pass in Algorithm~\ref{alg:c4vertex}. The local 4-cycle count for any $u$ vertex
adjacent to $v,y$ is simply $L(y)-1$ because it is in a 4-cycle with all other
neighbors of $v,y$. Therefore $\square(u)$ is updated in the second pass because
the final $L(y)$ count is needed. In Algorithm~\ref{alg:c4vertex}, we modify $L$
to be an array of 2-tuples to double buffer the counts of
$L(y)$.\footnote{Instead of copying the count, a single bit can be used to flag
when to reset the counts to zero in the first pass.} This retains the final
count for each $u$ while simultaneously allowing the $L(y)$ counts to be reset
to zero for the next $v$.

This clearly has the same complexity as Algorithm~\ref{alg:c4}; in particular,
Algorithm~\ref{alg:c4vertex} still uses just two passes over the graph. The
following result thus shows Theorem~\ref{thm:c4vertex}:

\begin{proposition}
  \label{prop:c4vertex_correct}
  Algorithm~\ref{alg:c4vertex} gives the correct count of all unique 4-cycles
  incident to each vertex.
\end{proposition}

\begin{proof}
  For any vertices $v, y$ with $y \prec v$, let $S(v,y)$ denote the set
  of vertices $u \in N(v) \cap N(y)$ where $u \prec v$.
  
  Consider a given vertex $x$ and a 4-cycle $C$ containing $x$. Let $v$ denote
  the $\prec$-largest vertex in $C$. There are a number of cases to consider.

  First, suppose $x = v$. In this case, we can write the cycle uniquely as $C =
  (x,u, y, u')$ where $u \prec u'$. The total number of such cycles is then
  $\sum_{y} \binom{ |S(x,y)|}{2}$. At Line~\ref{algc4vertex:2pathend} the value
  $L(y).\text{orig}$ is incremented for all vertices $y \prec v$ with $u \in
  S(x,y)$. At Line~\ref{algc4vertex:1stform} the running value of
  $L(y).\text{orig}$ gets added to $\square(x)$. Overall,
  Line~\ref{algc4vertex:1stform} accumulates a total value of $\sum_{y}
  \sum_{i=0}^{|S(x,y)|-1} i = \sum_{y} \binom{|S(x,y)|}{2}$.

  Second, suppose $x$ is antipodal to $v$ in the cycle, that is, $C =
  (v,u,x,u')$ where $u \prec u'$. This is very similar to the first case; the
  total number of such cycles gets counted by Line~\ref{algc4vertex:2ndform}.

  Finally, suppose $x$ is adjacent to $v$ in the cycle. In this case, we can
  write $C = (v,u,y,x)$. Here, for a given $v,y$, we can choose $u$ to be any
  vertex other than $x$ in $S(v,y)$. So the total number of such cycles is the
  sum of $(|S(v,y)| - 1)$ over all pairs $v,y$ with $x \in S(v,y)$. At the
  beginning of the loop at Line~\ref{algc4vertex:2ndpass_start}, the value
  $L(y).\text{copy}$ is equal to $|S(v,y)| - 1$; so this is precisely what is
  accumulated to $\square(x)$ at Line~\ref{algc4vertex:3rdform}.
\end{proof}

Our next algorithm computes the 4-cycle counts on each edge and is given in
Algorithm~\ref{alg:c4edge}. To avoid hash tables, we use an offset array $T$ to
map the start of an edge list in $M$ for each vertex. This array is simply a
prefix-sum array over the degrees of each vertex so a $T(v)$ element points to
the start of edges for $v$; note that $T$ can be assumed to be part of the
adjacency-list representation of $G$. This is similar to an approach
in~\cite{bib:burkhardt2021b} for identifying triangle neighbors, but comes at
the cost of greater algorithmic complexity.

\begin{algorithm}[H]
  \caption{Edge-local 4-cycle counting}\label{alg:c4edge}
  \begin{algorithmic}[1]
    \Require zero-initialized array $L$ of size $n$ indexed by
    $V$, where $L$ holds 2-tuples.
    \Require zero-initialized array $\square$ of size $m$
    \Require zero-initialized array $M$ of size $2m$
    \Require array $T$ of size $n$ defined by $T(v)=\sum_{v'<v}d(v')$
    \Comment{prefix sum array}
    \For{$v\in V$}
    \label{algc4edge:main_start}
      \For{$i=0$ \textbf{to} $d(v)-1$}
        \State set $u\coloneqq i^{th}$ neighbor of $v$        
        \If{$u\prec v$}
          \For{$j=0$ \textbf{to} $d(u)-1$}
            \State set $y\coloneqq j^{th}$ neighbor of $u$
            \If{$y\prec v$}
              \State set $L(y).\text{copy} \coloneqq L(y).\text{orig}$
              \Comment{save original $L(y)$ value}
              \label{algc4edge:g1}
              \State set $L(y).\text{orig} \coloneqq L(y).\text{orig} + 1$
            \EndIf
          \EndFor
        \EndIf
      \EndFor
      \For{$i=0$ \textbf{to} $d(v)-1$} \label{algc4edge:2ndpass_start}
        \State set $u\coloneqq i^{th}$ neighbor of $v$
        \If{$u\prec v$}
          \For{$j=0$ \textbf{to} $d(u)-1$}
            \State set $y\coloneqq j^{th}$ neighbor of $u$
            \If{$y \prec v$}
            \State set $M(T(v)+i)\coloneqq M(T(v)+i)$ + $L(y).\text{copy}$
            \label{algc4edge:g2}
            \State set $M(T(u)+j)\coloneqq M(T(u)+j)$ + $L(y).\text{copy}$
            \label{algc4edge:g3}
              \State set $L(y).\text{orig} \coloneqq 0$
            \EndIf
          \EndFor
        \EndIf
      \EndFor
    \EndFor
    \label{algc4edge:main_end}
    \For{$v\in V$}
      \For{$i=0$ \textbf{to} $d(v)-1$}
        \State set $u\coloneqq i^{th}$ neighbor of $v$
          \If{$u\prec v$}
            \State scan $N(u)$ to find index $j$ such that $v$ is the $j^{th}$
            neighbor of $u$.
            \label{algc4edge:find}
            \State set $\square(e) \coloneqq M(T(v)+i) + M(T(u)+j)$
            \label{algc4edge:add}
            \Comment{$e=\{v,u\}$}
          \EndIf
      \EndFor
    \EndFor
  \end{algorithmic}
\end{algorithm}

We use an array $M$ to maintain $C_4$ counts per edge, where $M$ is indexed in
the same ordering as the adjacency list representation. That is, the value
$M(T(v) + i)$ stores the corresponding value for the edge $vu$ where $u$ is the
$i^{th}$ neighbor of $v$. To see the correctness of Algorithm~\ref{alg:c4edge},
note that $\square(e)$ for edge $e=\{v,u\}$ is the sum of the counts stored in
$M$ for $(v,u)$ and $(u,v)$. These counts are added together in the final pass
over the edges.

The correctness and runtime of Algorithm~\ref{alg:c4edge} are given next by
Propositions~\ref{prop:c4edge_correct} and~\ref{prop:c4edge_time}, and
subsequently Theorem~\ref{thm:c4edge} holds.

\begin{proposition}
\label{prop:c4edge_correct}
  Algorithm~\ref{alg:c4edge} gives the correct count of all unique 4-cycles
  incident to each edge.
\end{proposition}

\begin{proof}
  For any vertices $v, y$ with $y \prec v$, let $S(v,y)$ denote the set of
  vertices $u \in N(v) \cap N(y)$ where $u \prec v$.
    
  Consider an edge $e$ and a 4-cycle $C$ containing $e$.  Let $v$ denote the
  $\prec$-largest vertex in $C$. There are two cases to consider.
  
  First, suppose $v$ is an endpoint of $e$. We can write $e = (v,u)$ and we can
  write $C$ uniquely as $(v,u, y,u')$. In this case, the total number of such
  cycles is the sum of $|S(v,y) - 1|$ over all pairs $v,y$ with $u \in S(v,y)$,
  since for a given $y$ we can choose $u'$ to be any element of $S(v,y)$ other
  than $u$.  At the beginning of the loop at Line~\ref{algc4edge:2ndpass_start},
  the value $L(y).\text{copy}$ is equal to $|S(v,y)| - 1$; so this is precisely
  what is accumulated to the counter $M(T(v) + i)$ at Line~\ref{algc4edge:g2}.

  Otherwise, suppose $v$ is not an endpoint of $e$. We can write $e = (u,y)$ and
  $C = (v,u,y,u')$. This is very similar to the first case; the total number of
  such cycles is accumulated to the counter $M(T(u) + j)$ at
  Line~\ref{algc4edge:g3}.

  Overall, every $C_4$ involving $e$ will be either counted at
  Line~\ref{algc4edge:g2} or Line~\ref{algc4edge:g3}. Later, at
  Line~\ref{algc4edge:add}, these will get added together to count $\square(e)$.
\end{proof}
  
\begin{proposition}
  \label{prop:c4edge_time}
  Algorithm~\ref{alg:c4edge} runs in $O(m \bar \delta)$ time.
\end{proposition}

\begin{proof}
  Clearly Lines~\ref{algc4edge:main_start}~--~\ref{algc4edge:main_end} take $O(m
  \bar \delta)$ time as in the previous algorithms. For the remaining steps, we
  implement Line~\ref{algc4edge:find} by a linear search over all neighbors of
  $u$ to find the index $j$. Summed over all edges, the complexity of these
  lines is given by
  \[
  \sum_{ vu \in E, u \prec v } d(u) \leq \sum_{vu \in E} \min\{d(v), d(u) \} =
  m\bar \delta.\qedhere
  \]
\end{proof}

Note that Algorithm~\ref{alg:c4edge} maintains the $M$ counter for all $2m$
bidirectional edges. By contrast, a hash-table approach would only need to
maintain $m$ counters (c.f.~Algorithm~\ref{alg:c4edge_hash}). The duplication of
edges allows us to avoid the cost of looking up edge indices, while the
hash-table data structure itself would require $O(m)$ additional memory to keep
the probing time in check.

In Appendix~\ref{app:hash} we give a simpler version of
Algorithm~\ref{alg:c4edge} based on hash tables, which may aid the exposition of
Algorithm~\ref{alg:c4edge}, but is significantly slower as demonstrated in
Section~\ref{sec:experiment}.

\section{\label{sec:opt} Saving constant factors in runtime via sorting}
In our algorithm descriptions, we have focused on algorithms which do not
require any rearrangement or sorting of the input data. If we are willing to
partially sort the adjacency lists, we can improve the runtime by some constant
factors.

For each vertex $v$, let us define $N^+(v) = \{ u \in N(v): u \succ v \}$ and
likewise $N^-(v) = \{u \in N(v):u \prec v \}$. We propose the following
preprocessing step for all our algorithms:

\begin{center}
  \defbox{Preprocessing Step:}{Sort the adjacency list $N^{+}(v)$ of each vertex
    $v$, such that the neighborhood begins with $N^{-}(v)$ (in arbitrary order),
    and then is followed by $N^{+}(v)$ (sorted in order of $\prec$).}
\end{center}

Afterwards, many other steps in the algorithm can be implemented more
efficiently. To illustrate, consider the basic 4-cycle counting algorithm,
listed below as Algorithm~\ref{alg:c4sort}.

\begin{algorithm}[H]
  \caption{4-cycle counting, with preprocessing}\label{alg:c4sort}
  \begin{algorithmic}[1]
    \Require zero-initialized array $L$ of size $n$ indexed by $V$
    \State Preprocess the adjacency lists of all vertices.
    \For{$v\in V$}
      \For{$u \in N^{-}(v)$}
      \label{algc4sort:u_begin}
        \State set $y$ to be the first neighbor of $u$
        (in its ordered adjacency list)
        \While{$y \neq v$} 
          \State set $\square(G)\coloneqq \square(G) + L(y)$
          \State set $L(y) \coloneqq L(y) + 1$
          \State set $y$ to be the next neighbor of $u$
        \EndWhile
      \EndFor
      \label{algc4sort:u_end}
      \For{$u \in N^{-}(v)$}
        \State set $y$ to be the first neighbor of $u$
        \While{$y \neq v$} 
          \State set $L(y) \coloneqq 0$
          \State set $y$ to be the next neighbor of $u$
        \EndWhile
      \EndFor
    \EndFor
  \end{algorithmic}
\end{algorithm}

Compare Lines~\ref{algc4sort:u_begin}~--~\ref{algc4sort:u_end} of
Algorithm~\ref{alg:c4sort} to the corresponding steps in Algorithm~\ref{alg:c4}:
instead of looping over all neighbors $y \in N(u)$, and then checking explicitly
if $y \prec v$, we only loop over the vertices $y \in N(u)$ with $y \prec v$
without actually checking the $\prec$ condition. Because of the sorted order of
the neighborhood $N(u)$, the vertices $y \in N^{-}(u)$ will come at the
beginning; since $u \prec v$, these will automatically satisfy $y \prec
v$. Then, the vertices $y \in N^{+}(u)$ come in sorted order, where we observe
that $v \in N^{+}(u)$.

Thus, these loops are faster by constant factors. We also observe that the
preprocessing step itself should be asymptotically faster than the other steps
of the algorithm. See Appendix~\ref{sec:arraysort_wallclock} for actual
wallclock running time.

\begin{proposition}
  The preprocessing step can be implemented in $O(m \log \bar \delta)$ time and
  $O(n)$ space.
\end{proposition}

\begin{proof}
  For each vertex $v$, we sort the neighborhood $N^{+}(v)$. If we write $d^+(x)
  = |N^{+}(x)|$ for brevity, then this takes $O(d^{+}(v)\log d^{+}(v))$ time and
  $O(n)$ space. Summing over all vertices, the runtime is given by

  \begin{align*}
    \sum_{v \in V} d^{+}(v) \log d^{+}(v) &= \sum_{v} \log d^{+}(v) \sum_{u \in
      N^{+}(v)} 1 = \sum_{vu \in E: v \prec u} \log d^{+}(v) \\
    &\leq \sum_{vu \in E: v \prec u} \log d(v) \leq \sum_{vu \in E} \log
    \min\{d(v), d(u) \}.
  \end{align*}

  At this point, we can apply Jensen's inequality to the concave-down function
  $x \mapsto \log x$. We have
  \[
  \frac{1}{m} \sum_{vu \in E} \log \min \{ d(v), d(u) \} \leq \log \Bigl(
  \frac{1}{m} \sum_{vu \in E} \min\{ d(v), d(u) \} \Bigr) = \log \bar \delta.
  \qedhere
\]
\end{proof}

A similar preprocessing step can be used in Algorithm~\ref{alg:c4vertex} and
Algorithm~\ref{alg:c4edge}. For the latter algorithm, we can save even more time
by using binary search to find the index $j$ at Line~\ref{algc4edge:find},
instead of linear search.

\section{\label{sec:experiment}Experiments: array vs hash table}
We ran experiments to investigate the running time of our algorithms using
arrays as opposed to using hash tables. In all tests, our array-based
implementations were at least $2\times$ faster, and on average $4\times$ -- $7
\times$ faster, than the corresponding hash-based implementation. We benchmarked
the runtimes on both synthetic grid graphs and real-world graphs. The graph
sizes ranged from about half a million to over 250 million vertices, with the
largest graph having over 3 billion edges and nearly 500 billion 4-cycles. The
graphs were simple, undirected and unweighted graphs with vertices labeled from
$0, \dots, n-1$ without gaps. The count of edges reported in our experiments
includes both directions of an edge.

We used hash table implementations available from the C++ standard library and
the Boost library~\cite[version 1.82]{bib:boost}; specifically,
\emph{std::unordered\_map} from C++, and from Boost the
\emph{boost::unordered\_map} and \emph{boost::unordered\_flat\_map}. The
standard hash table in C++ and Boost, both called ``unordered map'', use a
chaining strategy (linked lists) for collisions. The Boost ``unordered flat
map'' hash table uses open-addressing (probing) and is considered a very fast
hash table implementation. See also \cite{bib:hashbenchmark} for some more
extensive benchmarks on various hash table implementations.

We ran the experiments on a workstation with 28 Intel Xeon E5-2680 cores and 256
GB of RAM. The algorithm implementations were written in C++ and compiled with
``-std=C++11'' and ``-O3'' options enabled. The integer type for vertex labels
were \emph{32-bit unsigned int} (uint32).

All wallclock running times are given in
Appendix~\ref{app:runtime}. Table~\ref{tbl:avg_speedup} summarizes the average
speedup of our array method compared to using hash tables.

\begin{table}[H]
  \tablefontsize
  \centering
  \caption{Average speed-up using arrays versus hash tables over Grid graphs
    (grid) and Real-world graphs (real) from Table~\ref{tbl:graphdata} for
    4-cycle counts in total ($\square(G)$) and per vertex ($\square(v)$) and
    edge ($\square(e)$).}\label{tbl:avg_speedup}
  \begin{tabular}{lccccc}
    \toprule
    \multirow{2}{*}{Hash table} & \multirow{2}{*}{}
    & \multicolumn{3}{c}{Array speedup}
    & \multirow{2}{*}{\makecell{Overall Average\\ Speedup}} \\
    \cmidrule{3-5}
    && $\square(G)$ & $\square(v)$ & $\square(e)$ & {} \\
    \midrule
    \multirow{2}{*}{Boost flat}
    & (grid) & $4.4\times$ & $4.1\times$ & $5.6\times$
    & \multirow{2}{*}{$3.7\times$} \\
    \cmidrule{3-5}
    & (real) & $2.4\times$ & $2.6\times$ & $3.3\times$ \\
    \midrule
    \multirow{2}{*}{Boost}
    & (grid) & $6.4\times$ & $6.4\times$ & $5.7\times$
    & \multirow{2}{*}{$5.5\times$} \\
    \cmidrule{3-5}
    & (real) & $4.6\times$ & $4.9\times$ & $5.2\times$ \\
    \midrule
    \multirow{2}{*}{C++}
    & (grid) & $7.5\times$ & $8.1\times$ & $6.8\times$
    & \multirow{2}{*}{$7.3\times$} \\
    \cmidrule{3-5}
    & (real) & $6.6\times$ & $6.7\times$ & $8.0\times$ \\
    \bottomrule
  \end{tabular}
\end{table}

The graphs used in the benchmarks are tabulated in
Table~\ref{tbl:graphdata}. Note that the average degeneracy $\bar\delta$ is
significantly smaller than the core number $\kappa$ on the real-world
graphs. The first four are grid graphs we constructed. The remaining graphs are
real-world datasets from the Stanford Network Analysis Project
(SNAP)~\cite{bib:snapnets}.

\begin{table}[H]
\caption{Test Graphs}\label{tbl:graphdata} 
\tablefontsize
\centering
\begin{tabular}{lrrrrrr}
  \toprule
  Graph & $n$ (vertices) & $m$ (edges) & max degree
  & $\bar\delta$ & $\kappa$ & $\#C_4$ \\
  \midrule
  grid-$2^{18}\times 2^7$ & 33,554,432 & 133,693,184 & 4 & 3.98 & 2
  & 33,292,161 \\
  grid-$2^{19}\times 2^7$ & 67,108,864 & 267,386,624 & 4 & 3.98 & 2
  & 66,584,449 \\
  grid-$2^{20}\times 2^7$ & 134,217,728 & 534,773,504 & 4 & 3.98 & 2
  & 133,169,025 \\
  grid-$2^{21}\times 2^7$ & 268,435,456 & 1,069,547,264 & 4 & 3.98 & 2
  & 266,338,177 \\
  \midrule
  web-BerkStan & 685,230 & 13,298,940 & 84,230 & 38.6 & 201
  & 127,118,333,411 \\
  com-Youtube & 1,134,890 & 5,975,248 & 28,754 & 27.3 & 51
  & 468,774,021 \\
  as-Skitter & 1,696,415 & 22,190,596 & 35,455 & 36.5 & 111
  & 62,769,198,018 \\
  com-LiveJournal & 3,997,962 & 69,362,378 & 14,815 & 54.0 & 360
  & 26,382,794,168 \\
  com-Orkut & 3,072,441 & 234,370,166 & 33,313 & 143.2 & 253
  & 127,533,170,575 \\
  com-Friendster & 65,608,366 & 3,612,134,270 & 5,214 & 204.1 & 304
  & 465,803,364,346 \\
  \bottomrule
\end{tabular}
\end{table}

As the name suggests, the grid graphs are grids of $R$ rows by $C$ columns of
vertices linked together. Thus, a $3\times 3$ grid has a total of $9$ vertices
and $4$ copies of $C_4$. These graph have a fixed max degree of $4$, which
serves as a useful validation. We chose $R=2^j$ (with $j=18,19,20,21$) and
$C=2^7$ for the number of rows and columns, respectively. The number of
vertices, edges, and 4-cycles are given by the formulas:
\begin{align*}
& RC \tag{vertex count} \\
& (R-1)C+R(C-1) \tag{edge count} \\
& (R-1)(C-1) \tag{4-cycle count}
\end{align*}

Figure~\ref{fig:c4_compare} compares Algorithm~\ref{alg:c4} to the hash table
version given by Algorithm~\ref{alg:c4hash}. We emphasize
Algorithm~\ref{alg:c4hash} requires only a single pass over the graph data, as
opposed to the two passes for Algorithm~\ref{alg:c4}. Despite this, the
array-based method of Algorithm~\ref{alg:c4} outperforms
Algorithm~\ref{alg:c4hash} across all of the graphs on each of the library hash
tables in our study. A log-plot of the wallclock running time in seconds on grid
and real-world graphs is illustrated Figure~\ref{fig:c4_grid} and
Figure~\ref{fig:c4_snap}, respectively. Let $T_{\text{hash}}$ and
$T_{\text{array}}$ respectively denote the running time for hash table and
array-based implementations. If the array method is faster then
$T_{\text{hash}}/T_{\text{array}}$ is greater than one. These ratios are plotted
in Figures~\ref{fig:c4_ratio_grid},~\ref{fig:c4_ratio_snap}. It is clear that
Algorithm~\ref{alg:c4} using the array is faster on all counts than using a hash
table.

\begin{figure}[H]
  \centering
  \begin{subfigure}[t]{0.5\linewidth}
    \centering
    \resizebox{\linewidth}{!}{
      \includegraphics{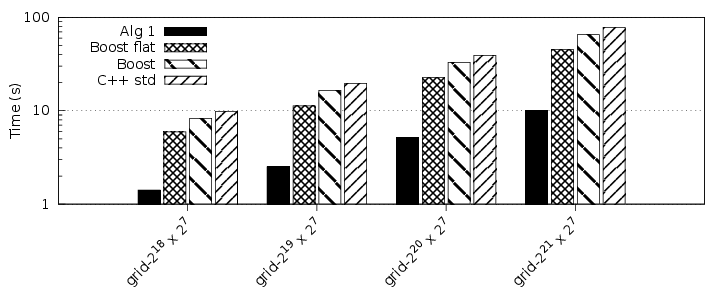}
    }
    \subcaption{Log-plot runtime on grid graphs.}
    \label{fig:c4_grid}
  \end{subfigure}%
  \begin{subfigure}[t]{0.5\linewidth}
    \centering
    \resizebox{\linewidth}{!}{
      \includegraphics{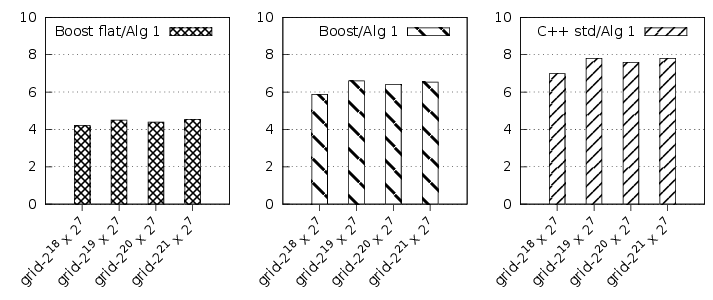}
    }
    \subcaption{Runtime ratio
      ($T_{\text{hash}}/T_{\text{array}}$) on grid graphs.}
    \label{fig:c4_ratio_grid}
  \end{subfigure}%
  \newline
  \begin{subfigure}[t]{0.5\linewidth}
    \centering
    \resizebox{\linewidth}{!}{
      \includegraphics{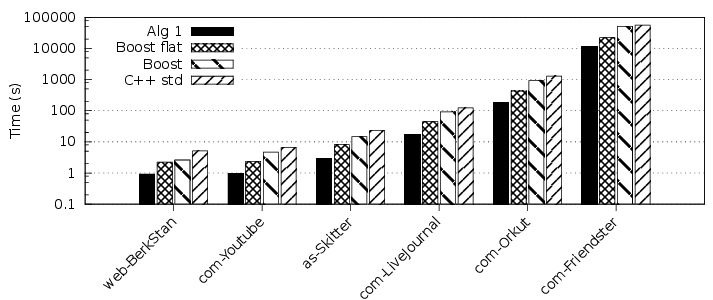}
    }
    \subcaption{Log-plot runtime on real-world graphs.}
    \label{fig:c4_snap} 
  \end{subfigure}%
  \begin{subfigure}[t]{0.5\linewidth}
    \centering
    \resizebox{\linewidth}{!}{
      \includegraphics{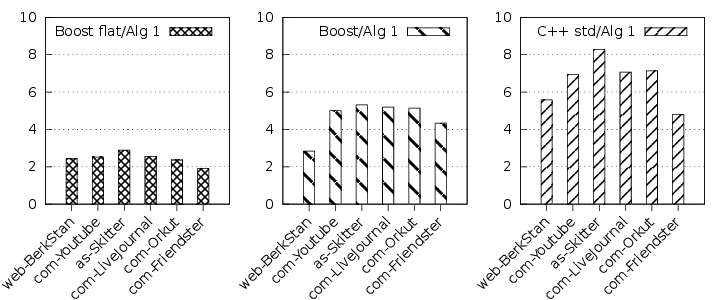}
    }
    \subcaption{Runtime ratio
      ($T_{\text{hash}}/T_{\text{array}}$) on real-world graphs.}
    \label{fig:c4_ratio_snap}
  \end{subfigure}
  \caption{Array (Algorithm~\ref{alg:c4}) versus hash table
    (Algorithm~\ref{alg:c4hash}) for counting
    $\square(G)$.}\label{fig:c4_compare}
\end{figure}

Next we show in Figure~\ref{fig:c4vertex_compare} that our array-based
Algorithm~\ref{alg:c4vertex} for computing 4-cycle counts on each vertex also
outperforms the corresponding hash table implementation given in
Algorithm~\ref{alg:c4vertex_hash}. Unlike our hash-based
Algorithm~\ref{alg:c4hash}, we cannot skip the second pass over the graph in
Algorithm~\ref{alg:c4vertex_hash}. We remind the reader that the current
methods~\cite{bib:sanei2018, bib:wang2019, bib:zhou2021, bib:wang2022} also take
more than one pass. For Algorithm~\ref{alg:c4vertex} we implemented the array
$L$ to hold pairs of integers in order to double buffer the counts of $y$ in the
degree-oriented path $(v,u,y)$. This is more cache-friendly than allocating two
separate arrays. We also tested using pairs holding an integer and Boolean,
where the Boolean value was used to flag when to reset the counts in the first
pass. But we did not observe significant performance gain despite this approach
being more space-efficient.

\begin{figure}[H]
  \centering
  \begin{subfigure}[t]{0.5\linewidth}
    \centering
    \resizebox{\linewidth}{!}{
      \includegraphics{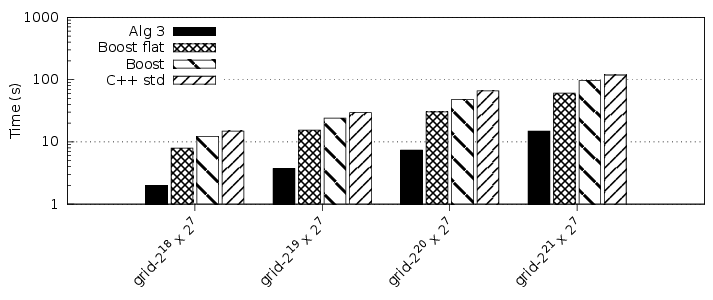}
    }
    \subcaption{Log-plot runtime on grid graphs.}
    \label{fig:c4vertex_grid}
  \end{subfigure}%
  \begin{subfigure}[t]{0.5\linewidth}
    \centering
    \resizebox{\linewidth}{!}{
      \includegraphics{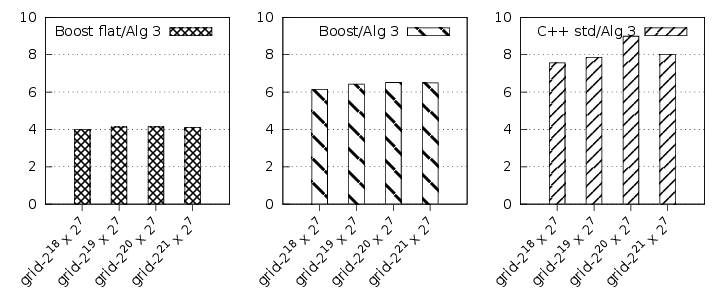}
    }
    \subcaption{Runtime ratio
      ($T_{\text{hash}}/T_{\text{array}}$) on grid graphs.}
    \label{fig:c4vertex_ratio_grid}
  \end{subfigure}%
  \newline
  \begin{subfigure}[t]{0.5\linewidth}
    \centering
    \resizebox{\linewidth}{!}{
      \includegraphics{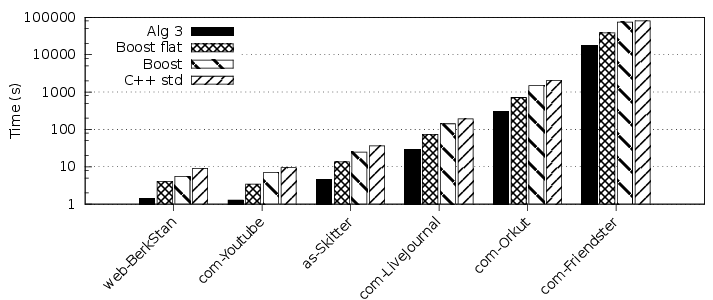}
    }
    \subcaption{Log-plot runtime on real-world graphs.}
    \label{fig:c4vertex_snap} 
  \end{subfigure}%
  \begin{subfigure}[t]{0.5\linewidth}
    \centering
    \resizebox{\linewidth}{!}{
      \includegraphics{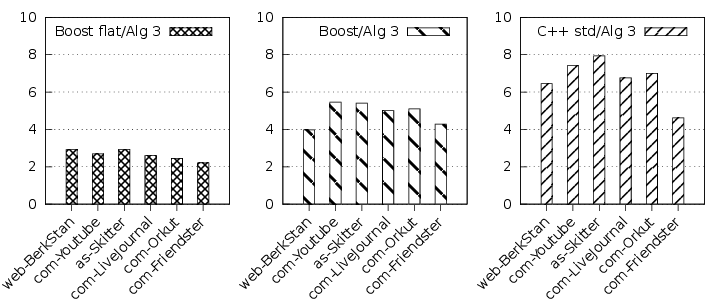}
    }
    \subcaption{Runtime ratio
      ($T_{\text{hash}}/T_{\text{array}}$) on real-world graphs.}
    \label{fig:c4vertex_ratio_snap}
  \end{subfigure}
  \caption{Array (Algorithm~\ref{alg:c4vertex}) versus hash table
    (Algorithm~\ref{alg:c4vertex_hash}) for counting $\square(v)$.}
  \label{fig:c4vertex_compare}
\end{figure}

Finally, Figure~\ref{fig:c4edge_compare} gives a comparison of runtime
performance for edge-local 4-cycle counting, showing that our array-based
Algorithm~\ref{alg:c4edge} is faster than the corresponding hash table
implementation given in Algorithm~\ref{alg:c4edge_hash}. The timing for
Algorithm~\ref{alg:c4edge} includes the set-up time to construct the prefix-sum
array $T$. The hash table on edges uses a 64-bit integer key, and for each edge
$\{u, v \}$ we stored the endpoint that sorted first in the higher-order 32 bits
of the key and the other endpoint in the lower remaining bits.

\begin{figure}[H]
  \centering
  \begin{subfigure}[t]{0.5\linewidth}
    \centering
    \resizebox{\linewidth}{!}{
      \includegraphics{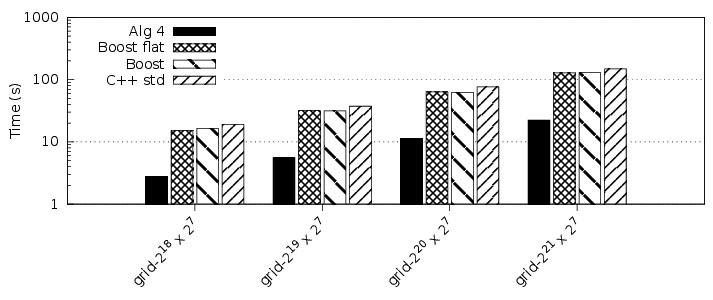}
    }
    \subcaption{Log-plot runtime on grid graphs.}
    \label{fig:c4edge_grid}
  \end{subfigure}%
  \begin{subfigure}[t]{0.5\linewidth}
    \centering
    \resizebox{\linewidth}{!}{
      \includegraphics{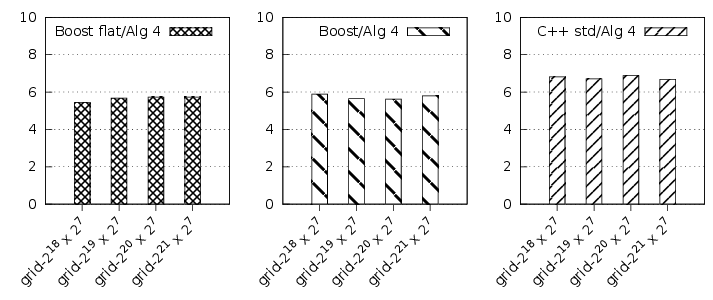}
    }
    \subcaption{Runtime ratio
      ($T_{\text{hash}}/T_{\text{array}}$) on grid graphs.}
    \label{fig:c4edge_ratio_grid}
  \end{subfigure}%
  \newline
  \begin{subfigure}[t]{0.5\linewidth}
    \centering
    \resizebox{\linewidth}{!}{
      \includegraphics{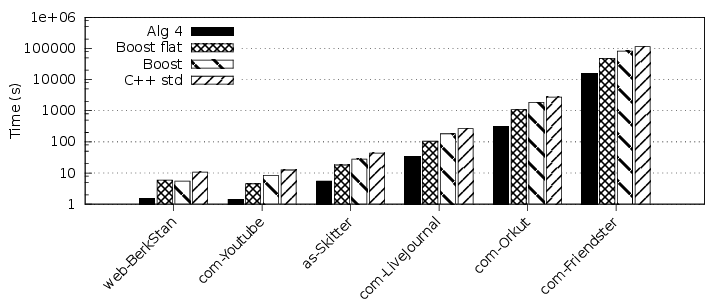}
    }
    \subcaption{Log-plot runtime on real-world graphs.}
    \label{fig:c4edge_snap} 
  \end{subfigure}%
  \begin{subfigure}[t]{0.5\linewidth}
    \centering
    \resizebox{\linewidth}{!}{
      \includegraphics{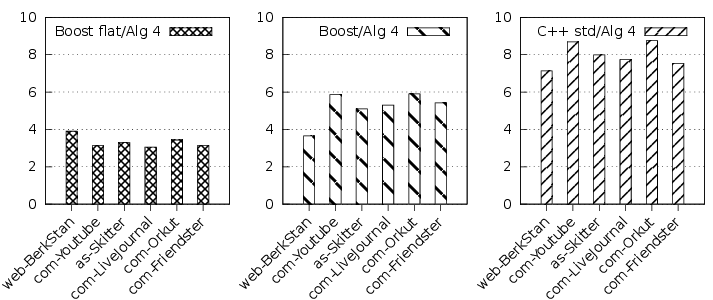}
    }
    \subcaption{Runtime ratio
      ($T_{\text{hash}}/T_{\text{array}}$) on real-world graphs.}
    \label{fig:c4edge_ratio_snap}
  \end{subfigure}
  \caption{Array (Algorithm~\ref{alg:c4edge})
    versus hash table (Algorithm~\ref{alg:c4edge_hash}) for counting
    $\square(e)$.}\label{fig:c4edge_compare}
\end{figure}

These empirical results demonstrate that our array-based algorithms are not only
very simple but also faster than hash table methods. The practical runtime can
be further improved using our Algorithm~\ref{alg:c4sort} if a pre-step sorting
pass is permitted. For the graphs in this study Algorithm~\ref{alg:c4sort} is on
average $1.7\times$ faster than our primary Algorithm~\ref{alg:c4}. The results
are illustrated in Figure~\ref{fig:c4sort_compare} using only the algorithm
wallclock time. The actual timings including the time for the pre-step sorting
pass are reported in Appendix~\ref{sec:arraysort_wallclock}. We used the sorting
algorithm from the standard C++ library and sort only the higher-degree
neighbors in $N^{+}(v)$.

\begin{figure}[H]
  \centering
  \begin{subfigure}[t]{0.5\linewidth}
    \centering
    \resizebox{\linewidth}{!}{
      \includegraphics{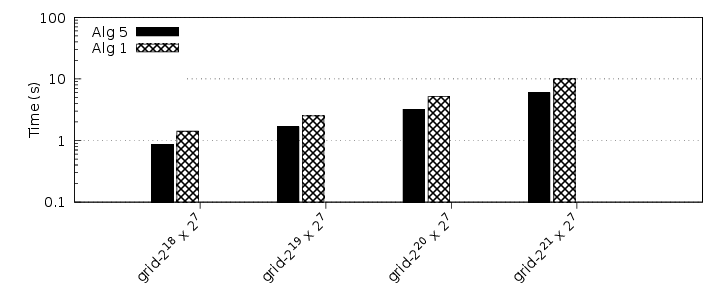}
    }
    \subcaption{Log-plot runtime on grid graphs.}
    \label{fig:c4sort_grid}
  \end{subfigure}%
  \begin{subfigure}[t]{0.5\linewidth}
    \centering
    \resizebox{0.5\linewidth}{!}{
      \includegraphics{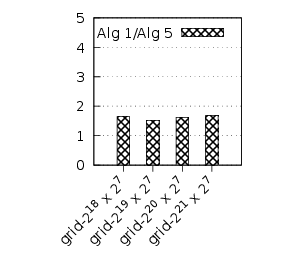}
    }
    \subcaption{Runtime ratio
      ($T_{\text{hash}}/T_{\text{array}}$) on grid graphs.}
    \label{fig:c4sort_ratio_grid}
  \end{subfigure}%
  \newline
  \begin{subfigure}[t]{0.5\linewidth}
    \centering
    \resizebox{\linewidth}{!}{
      \includegraphics{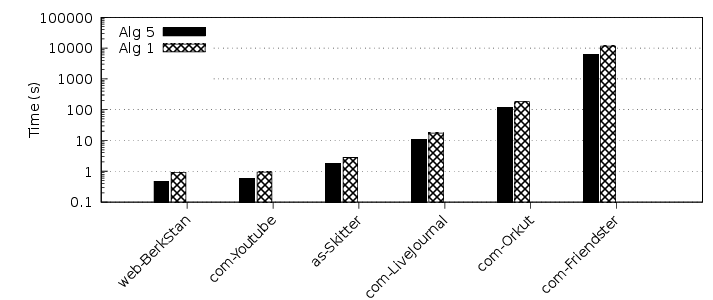}
    }
    \subcaption{Log-plot runtime on real-world graphs.}
    \label{fig:c4sort_snap} 
  \end{subfigure}%
  \begin{subfigure}[t]{0.5\linewidth}
    \centering
    \resizebox{0.5\linewidth}{!}{
      \includegraphics{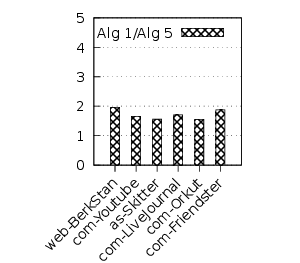}
    }
    \subcaption{Runtime ratio
      ($T_{\text{hash}}/T_{\text{array}}$) on real-world graphs.}
    \label{fig:c4sort_ratio_snap}
  \end{subfigure}
  \caption{Algorithm~\ref{alg:c4sort} (sort-based) versus
    Algorithm~\ref{alg:c4} for counting $\square(G)$.}
  \label{fig:c4sort_compare}
\end{figure}

We expect comparable speedup for our other algorithms
(Algorithm~\ref{alg:c4vertex} and Algorithm~\ref{alg:c4edge}). There are many
optimizations that could potentially improve the performance of either the array
or hash table based implementations. Our motivation here is to provide
supporting evidence that our simple array-based approach is an appealing
alternative to the current hash-based methods.

\bibliographystyle{abbrv}
\bibliography{fourcycle_counting}

\appendix
\section{\label{app:hash}Hash table derivatives of Algorithm~\ref{alg:c4}}
In our Algorithm~\ref{alg:c4}, the $L$ table is a flat array of integers in
order to ensure deterministic performance. In other approaches,
e.g. \cite{bib:wang2019, bib:wang2022}, a hash table is used instead. On one
hand, a hash table may use less memory on average (depending on the actual
number of vertices indexed); on the other hand, a hash table has higher overhead
due to collisions, etc.

There is a particularly clean variant of Algorithm~\ref{alg:c4} using hash
tables, which we describe as follows:

\begin{algorithm}[H]
  \caption{4-cycle counting with hash tables}\label{alg:c4hash}
  \begin{algorithmic}[1]
    \For{$v\in V$}
      \State Initialize a zero-filled hash-table $H$ indexed by $V$
      \For{$u\in N(v)$}
        \If{$u\prec v$}
          \For{$y\in N(u)$}
            \If{$y\prec v$}
              \State set $\square(G)\coloneqq \square(G) + H(y)$
              \State set $H(y) \coloneqq H(y) + 1$
            \EndIf
          \EndFor
        \EndIf
      \EndFor
    \EndFor
  \end{algorithmic}
\end{algorithm}

This requires just a single pass for each vertex $v$, in contrast to previous
algorithms~\cite{bib:sanei2018, bib:wang2019, bib:zhou2021, bib:wang2022};
instead of explicitly zeroing out the hash table $H$, we can simply discard it
and re-initialize for the next vertex $v$.

Along similar lines, replacing the $L$ array with a hash table in
Algorithm~\ref{alg:c4vertex} yields the following hash-based, vertex-local
4-cycle counting algorithm.

\begin{algorithm}[H]
  \caption{Hash-based Vertex-local 4-cycle counting}\label{alg:c4vertex_hash}
  \begin{algorithmic}[1]
    \Require zero-initialized array $\square$ of size $n$ indexed by $V$
    \For{$v\in V$}
      \State Initialize a zero-filled hash table $H$ indexed by $V$
      \For{$u\in N(v)$}
        \If{$u\prec v$}
          \For{$y\in N(u)$}
            \If{$y\prec v$}
              \State set $\square(v) := \square(v) + H(y)$
              \State set $\square(y) := \square(y) + H(y)$
              \State set $H(y) \coloneqq H(y) + 1$
            \EndIf
          \EndFor
        \EndIf
      \EndFor
      \For{$u\in N(v)$}
        \If{$u\prec v$}
          \For{$y\in N(u)$}
            \If{$y\prec v$}
              \State set $\square(u)\coloneqq \square(u) + H(y)-1$
            \EndIf
          \EndFor
        \EndIf
      \EndFor
    \EndFor
  \end{algorithmic}
\end{algorithm}

We can extend Algorithm~\ref{alg:c4vertex} to compute the 4-cycle counts on each
edge; see Algorithm~\ref{alg:c4edge_hash} for details. This does \emph{not}
directly give Theorem~\ref{thm:c4edge}, due to the use of a randomized data
structure. We emphasize that whenever the hash table on edges $H_e$ is accessed
we sort the edge endpoints to get the correct count for the unordered endpoint
pair $\{x,y\}$.

\begin{algorithm}[H]
  \caption{Hash-based Edge-local 4-cycle counting}\label{alg:c4edge_hash}
  \begin{algorithmic}[1]
    \Require zero-initialized hash table $H_v$ of size $n$
    \Require zero-initialized hash table $H_e$ of size $m$
    \For{$v\in V$}
      \For{$u\in N(v)$}
        \If{$u\prec v$}
          \For{$y\in N(u)$}
            \If{$y\prec v$}
              \State set $H_v(y) \coloneqq H_v(y) + 1$
            \EndIf
          \EndFor
        \EndIf
      \EndFor
      \For{$u\in N(v)$}
        \If{$u\prec v$}
          \For{$y\in N(u)$}
            \If{$y\prec v$}
              \State set $H_e(\{v,u\})\coloneqq H_e(\{v,u\}) + H_v(y)-1$
              \label{algc4edge_hash:1stform}
              \State set $H_e(\{u,y\})\coloneqq H_e(\{u,y\}) + H_v(y)-1$
              \label{algc4edge_hash:2ndform}
            \EndIf
          \EndFor
        \EndIf
      \EndFor
    \EndFor
  \end{algorithmic}
\end{algorithm}

In Algorithm~\ref{alg:c4edge_hash} we update $H_e(\{v,u\}),H_e(\{u,y\})$ in the
second iteration for similar reasons why we updated $\square(u)$ in the second
iteration of Algorithm~\ref{alg:c4vertex}. Namely, since $u$ is the intersecting
vertex in the $(v,u,y)$ path for a 4-cycle, then $u$ is in $H_v(y)-1$ cycles
involving the endpoints $v,y$ and hence the final count of $H_v(y)$ is needed.

\section{\label{app:runtime}Wallclock times}
\subsection{Algorithm~\ref{alg:c4} and hash-based Algorithm~\ref{alg:c4hash} for
  counting $\square(G)$}
\begin{table}[H]
  \tablefontsize
  \caption{Wallclock time in seconds for counting
    $\square(G)$.}\label{tbl:c4_wallclock}
  \sisetup{round-mode=places, round-precision=2, group-separator={,}}
  \centering
  \begin{tabular}{l *4{S[table-format=5.2]}}
    \toprule
    && \multicolumn{3}{c}{Algorithm~\ref{alg:c4hash}} \\
    \cmidrule{3-5}
    {} & {Algorithm~\ref{alg:c4}} & {Boost Flat} & {Boost} & {C++ std} \\
    \midrule
    grid-$2^{18}\times 2^7$ & 1.409898 & 5.946442 & 8.282919 & 9.846832 \\
    grid-$2^{19}\times 2^7$ & 2.507535 & 11.259669 & 16.530271 & 19.573684 \\
    grid-$2^{20}\times 2^7$ & 5.156348 & 22.614043 & 33.021225 & 39.107797 \\
    grid-$2^{21}\times 2^7$ & 10.001205 & 45.213457 & 65.358424 & 78.068043 \\
    \midrule
    web-BerkStan & 0.919170 & 2.232725 & 2.610056 & 5.127991 \\
    com-Youtube & 0.941304 & 2.371421 & 4.714385 & 6.532483 \\
    as-Skitter & 2.797216 & 8.092537 & 14.865716 & 23.177490 \\
    com-LiveJournal & 17.520293 & 44.698306 & 91.050951 & 123.690076 \\
    com-Orkut & 181.799631 & 430.047626 & 932.826652 & 1296.495744 \\
    com-Friendster & 11680.592756 & 22339.081823 & 50641.689680 & 56099.331270
    \\
    \bottomrule
  \end{tabular}
\end{table}

\subsection{Algorithm~\ref{alg:c4vertex} and hash-based
  Algorithm~\ref{alg:c4vertex_hash} for counting $\square(v)$}

\begin{table}[H]
  \tablefontsize
  \caption{Wallclock time in seconds for counting
    $\square(v)$.}\label{tbl:c4vertex_wallclock}
  \sisetup{round-mode=places, round-precision=2, group-separator={,}}
  \centering
  \begin{tabular}{l *4{S[table-format=5.2]}}
    \toprule
    && \multicolumn{3}{c}{Algorithm~\ref{alg:c4vertex_hash}} \\
    \cmidrule{3-5}
    {} & {Algorithm~\ref{alg:c4vertex}} & {Boost Flat} & {Boost} & {C++ std} \\
    \midrule
    grid-$2^{18}\times 2^7$ & 1.974360 & 7.899200 & 12.113139 & 14.927022 \\
    grid-$2^{19}\times 2^7$ & 3.735104 & 15.407606 & 23.980264 & 29.336168 \\
    grid-$2^{20}\times 2^7$ & 7.352238 & 30.347825 & 47.848993 & 66.066654 \\
    grid-$2^{21}\times 2^7$ & 14.787558 & 60.506046 & 95.912017 & 118.372449 \\
    \midrule
    web-BerkStan & 1.397283 & 4.076517 & 5.555774 & 9.020273 \\
    com-Youtube & 1.273454 & 3.426832 & 6.955773 & 9.448993 \\
    as-Skitter & 4.590072 & 13.479278 & 24.813039 & 36.421807 \\
    com-LiveJournal & 28.340886 & 73.624984 & 142.095220 & 191.508372 \\
    com-Orkut & 292.130614 & 713.071135 & 1490.447212 & 2044.297366 \\
    com-Friendster & 17370.654695 & 38427.779236 & 74335.696056 & 80277.221946 \\
    \bottomrule
  \end{tabular}
\end{table}

\subsection{Algorithm~\ref{alg:c4edge} and hash-based
  Algorithm~\ref{alg:c4edge_hash} for counting $\square(e)$}

\begin{table}[H]
  \tablefontsize
  \caption{Wallclock time in seconds for counting
    $\square(e)$.}\label{tbl:c4edge_wallclock}
  \sisetup{round-mode=places, round-precision=2, group-separator={,}}
  \centering
  \begin{tabular}{l *4{S[table-format=6.2]}}
    \toprule
    && \multicolumn{3}{c}{Algorithm~\ref{alg:c4edge_hash}} \\
    \cmidrule{3-5}
    {} & {Algorithm~\ref{alg:c4edge}} & {Boost Flat} & {Boost} & {C++ std} \\
    \midrule
    grid-$2^{18}\times 2^7$ & 2.796349 & 15.179904 & 16.464479 & 19.040180 \\
    grid-$2^{19}\times 2^7$ & 5.575090 & 31.646897 & 31.512673 & 37.383902 \\
    grid-$2^{20}\times 2^7$ & 11.079316 & 63.484692 & 62.230866 & 76.161981 \\
    grid-$2^{21}\times 2^7$ & 22.275985 & 127.706312 & 128.991308 & 148.667144 \\
    \midrule
    web-BerkStan & 1.501359 & 5.858236 & 5.482321 & 10.710947 \\
    com-Youtube & 1.440790 & 4.524086 & 8.451709 & 12.505601 \\
    as-Skitter & 5.494119 & 18.113570 & 28.014767 & 43.918403 \\
    com-LiveJournal & 34.315179 & 104.598719 & 181.842732 & 265.444243 \\
    com-Orkut & 314.220578 & 1078.084563 & 1852.628516 & 2751.603178 \\
    com-Friendster & 15211.269069 & 47820.847604 & 82507.836505 & 114516.718534
    \\
    \bottomrule
  \end{tabular}
\end{table}

\subsection{\label{sec:arraysort_wallclock}Sort-based Algorithm~\ref{alg:c4sort}
  and Algorithm~\ref{alg:c4} for counting $\square(G)$}

\begin{table}[H]
  \tablefontsize
  \caption{Wallclock time in seconds for counting
    $\square(G)$.}\label{tbl:c4sort_wallclock}
  \sisetup{round-mode=places, round-precision=2, group-separator={,}}
  \centering
  \begin{tabular}{l *3{S[table-format=5.2]}}
    \toprule
    & \multicolumn{2}{c}{Algorithm~\ref{alg:c4sort}} \\
    \cmidrule{2-3}
    {} & {pre-step sort} & {run} & {Algorithm~\ref{alg:c4}} \\
    \midrule
    grid-$2^{18}\times 2^7$ & 0.790889 & 0.856175 & 1.409898 \\
    grid-$2^{19}\times 2^7$ & 1.500563 & 1.653989 & 2.507535 \\
    grid-$2^{20}\times 2^7$ & 3.000943 & 3.184887 & 5.156348 \\
    grid-$2^{21}\times 2^7$ & 7.099967 & 5.922667 & 10.001205 \\
    \midrule
    web-BerkStan & 0.211910 & 0.469464 & 0.919170 \\
    com-Youtube & 0.134649 & 0.568863 & 0.941304 \\
    as-Skitter & 0.320545 & 1.798782 & 2.797216 \\
    com-LiveJournal & 1.436337 & 10.254448 & 17.520293 \\
    com-Orkut & 5.121030 & 115.621703 & 181.799631 \\
    com-Friendster & 105.689000 & 6026.564089 & 11680.592756 \\
    \bottomrule
  \end{tabular}
\end{table}

\end{document}